\documentclass[12pt, a4paper]{amsart}

\usepackage{amsmath,amssymb,amsfonts}
\usepackage{bbold}
\usepackage{epic,eepic}
\usepackage{enumerate}

\usepackage{amsthm}
\usepackage[all]{xy}
\usepackage{caption}
\usepackage{stmaryrd} 
\usepackage[dvipdfmx]{graphicx,color} 
\usepackage{tikz}
\usetikzlibrary{trees}

\usepackage{url}

\usepackage{array,multirow} 

\usepackage{geometry}
\usepackage{listings} 

\theoremstyle{plain}
\newtheorem{thm}{Theorem}[section]
\newtheorem{prop}[thm]{Proposition}

\theoremstyle{definition}
\newtheorem{defn}[thm]{Definition}
\newtheorem{exmp}[thm]{Example}

\numberwithin{figure}{section}

\numberwithin{table}{section}

\newcommand{\lspace} {
  \vspace{0.8\baselineskip}
}

\newcommand{\single}[1]{
  \langle  #1 \rangle 
}

\newdir{ >}{{}*!/-5pt/@{>}}

\definecolor{arrowred}{rgb}{0,0,0} 
\newcommand{\newword}[1]{\textbf{\textit{#1}}}

\newcommand{\textred}[1]{#1}

\numberwithin{equation}{section}

\newcommand{\head}{\mathbf{head}}
\newcommand{\tail}{\mathbf{tail}}

\mathchardef\mhyphen="2D  

\title [Discrete Signature and its Application to Finance]{Discrete Signature\\and its Application to Finance}

\thanks{This work was supported by JSPS KAKENHI Grant Number 18K01551.}

\author[T. Adachi, Y. Naritomi]{Takanori Adachi and Yusuke Naritomi}

\address{Graduate School of Management, Tokyo Metropolitan University}
\email{Takanori Adachi <taka.adachi@tmu.ac.jp>}

\address{Graduate School of Management, Tokyo Metropolitan University}
\email{Yusuke Naritomi <naritomi-yusuke@ed.tmu.ac.jp>}

\date{\today}

\begin{document}

\maketitle




\begin{abstract}
Signatures, 
one of the key concepts of rough path theory, 
have recently gained prominence as a means to find appropriate feature sets 
in machine learning systems.

In this paper, 
in order to compute signatures directly from discrete data without going through the transformation to continuous data,
we introduced a discretized version of signatures, called "flat discrete signatures".
We showed that the flat discrete signatures can represent the quadratic variation that has a high relevance in financial applications.
We also introduced the concept of "discrete signatures" that is a generalization of "flat discrete signatures".
This concept is defined to reflect the fact that data closer to the current time is more important than older data, 
and is expected to be applied to time series analysis.

As an application of discrete signatures, 
we took up 
a stock market related problem 
and succeeded in performing 
a good estimation 
with fewer data points than before.
\end{abstract}

\section{Introduction}
\label{sec:introduction}

The signature,
one of the key concepts of rough path theory,
is recently considered as
a means to find an appropriate feature set in machine learning systems
\cite{CK_2016}.
It may become a powerful tool when combining with traditional machine learning techniques such as deep learning.
In this paper,
we introduce a new concept called discrete signatures, and apply it to some financial problems.

In Section \ref{sec:discSig},
we introduce a concept of flat discrete signatures
that is a simple discretization of the traditional signatures defined in
\cite{LCL2007},
but with the
$\head$-$\tail$ transformation
that is an enlargement method of the underlying alphabet set.
We show that the
$\head$-$\tail$ transformation,
just like
the lead-lag transformation of streams,
provides
the quadratic variation of 
any component of the original process.
This is important since
the quadratic variation has a high relevance in financial applications.
%
When applying flat discrete signatures to time-series analysis,
we often encounter the necessity of treating data closer to the present time 
as more important than older data.
In order to address this problem,
we generalize flat discrete signatures
to reflect the fact.
The resulting version is called discrete signatures.

In Section \ref{sec:ImpleDS},
we will make a brief explanation about how we implement the signatures.
Actually, an implementation of signatures was made by 
Patrick Kidger 
and Terry Lyons 
as a Python-usable library called Signatory workable with PyTorch, 
which is written in C++ \cite{kidger2021signatory}.
We will present yet another, 
but a very simple implementation using Python 
by adopting discrete signatures.

In Section \ref{sec:appFin},
as an example of applications of discrete signature to finance,
we consider the problem of judging 
whether a given price-shares process is of the morning or of the afternoon session in 
Tokyo Stock Exchange.
We make a logistic regression with components of discrete signatures as features or explanatory variables.
Then we will see that our result is as good as the regression with the whole raw data set
with much fewer data points.

\section{Discrete Signature}
\label{sec:discSig}

Throughout this paper, we fix the discrete time domain
\begin{equation}
\mathcal{T}
	:=
\{
	t_0, t_1, t_2, \cdots 
\} 
\end{equation}
with
\[
0 = t_0 < t_1 < \cdots < t_n < t_{n+1} < \cdots
\]
and a \textit{discrete} path 
$X$
in
$\mathbb{R}^d$
for some fixed positive integer
$d$,
which can be written like
\begin{equation}
\xymatrix@C=20 pt@R=25 pt{
    \mathcal{T}
        \ar @{->}^{\textred{X}} [rr]
        \ar @{}_{ \mathrel{ \rotatebox[origin=c]{90}{$\in$} } } @<+6pt> [d]
&&
    \mathbb{R}^d
        \ar @{}_{ \mathrel{ \rotatebox[origin=c]{90}{$\in$} } } @<+6pt> [d]
&
\\
    t
        \ar @{|->} [rr]
&&
    X_t
        \ar @{}_-{=} @<+6pt> [r]
&
	(X_t^1, \cdots, X_t^d) .
}
\end{equation}


\begin{defn}{[Word]}
\begin{align}
I
	&:=
\{ 1, 2, \cdots, d \} ,
	\\
I^*
	&:=
\bigcup_{k = 0}^{\infty}
	I^k .
\end{align}
We call 
an element of
$I$
an \newword{alphabet}
and
an element of
$I^*$
a \newword{word} 
or a \newword{multi-index}.

The unique element of
$I^0$
is denoted by
$\lambda$,
which is the word with length $0$, or the \newword{empty word}.

The 
\newword{concatenation}
of two words
$u \in I^j$
and
$v \in I^k$,
denoted by
$u \otimes v$,
is the word
$w \in I^{j+k}$
defined by
for
$i \in 
\{1, 2, \cdots, j+k\}
$,
\begin{equation}
w_i := \begin{cases}
	u_i
		\quad \textrm{if} \quad
	1 \le i \le j,
\\
	v_{i-j}
		\quad \textrm{if} \quad
	j+1 \le i \le j+k,
\end{cases}
\end{equation}
where
$w_i$
stands for
$w(i)$.
\end{defn}

We usually focus a finite subset of 
$I^*$ such as
\begin{equation}
\label{eq:IleK}
I^{\le k}
  :=
\bigcup_{\ell = 0}^k  I^{\ell} .
\end{equation}


First, we will see the traditional definition of (continuous) signatures.

\begin{defn}{[Signature \cite{LCL2007}]}
\label{defn:contSig}
Let
$\mathbb{R}_+$
be the continuous time domain
starting from $0$,
$I$
be an alphabet set
and
$
\tilde{X} : 
\mathbb{R}_+
	\to
\mathbb{R}^I
$
be a path.
Let
$
a, b \in
\mathbb{R}_+
$
with
$a < b$.
\begin{enumerate}
\item
For
$w \in I^*$,
$
S(\tilde{X})_{a, b}^w
	\in
\mathbb{R}
$
is defined inductively by
\begin{align}
S(\tilde{X})_{a, b}^{\lambda}
	&:=
1,
	\\
S(\tilde{X})_{a, b}^{w \otimes i}
	&:=
\int_a^b
	S(\tilde{X})_{a, t}^{w}
	d \tilde{X}^i_t
			\quad (\textrm{for } i \in I) ,
\end{align}
where
\begin{equation}
	d \tilde{X}^i_t
		:=
	\dot{\tilde{X}}^i_t
	dt .
\end{equation}

\item
The (traditional) \newword{signature} of
$\tilde{X}$
over
$[a, b]$
is a function
$
S(\tilde{X})_{a, b}
	:
I^* \to \mathbb{R}
$
defined by
\begin{equation}
S(\tilde{X})_{a, b}(w)
	:=
S(\tilde{X})_{a, b}^w
\end{equation}
for
$w \in I^*$.

\end{enumerate}
\end{defn}

Because we have a discrete path
$
X :
\mathcal{T}
	\to
\mathbb{R}^d
$,
we have to convert 
$X$
to
an appropriate continuous time path
$\tilde{X}$
before computing its signature.

One of the natural ways to accomplish this is an interpolation.
If we adopt the linear interpolation to fill the values 
between
$t_n$
and
$t_{n+1}$,
we have
for
$
t_n \le t < t_{n+1}
$
\begin{equation}
\tilde{X}^i_t
	:=
\frac{
	X^i_{t_n}
	(t_{n+1} - t)
		+
	X^i_{t_{n+1}}
	(t - t_n)
}{
	t_{n+1} - t_n
} .
\end{equation}
Then for
$
t_n \le t < t_{n+1}
$,
\begin{equation}
\dot{\tilde{X}}^i_t
	=
\frac{
	X^i_{t_{n+1}}
		-
	X^i_{t_n}
}{
	t_{n+1} - t_n
} .
\end{equation}
Therefore,
for
$m, n \in \mathbb{N}$
with
$m < n$,
and
$i \in I$ ,
\begin{align*}
S(\tilde{X})_{t_m, t_n}^{w \otimes i}
	&=
\int_{t_m}^{t_n}
	S(\tilde{X})_{t_m, t}^{w}
	\dot{\tilde{X}}^i_t
	d t
	=
\sum_{\ell=m}^{n-1}
\int_{t_{\ell}}^{t_{\ell+1}}
	S(\tilde{X})_{t_m, t}^{w}
	\dot{\tilde{X}}^i_t
	d t
	\\&=
\sum_{\ell=m}^{n-1}
\frac{
	X^i_{t_{\ell+1}} - X^i_{t_{\ell}}
}{
	t_{\ell+1} - t_{\ell}
} 
\int_{t_{\ell}}^{t_{\ell+1}}
	S(\tilde{X})_{t_m, t}^{w}
	d t
	=
\sum_{\ell=m}^{n-1}
\frac{
	X^i_{t_{\ell+1}} - X^i_{t_{\ell}}
}{
	t_{\ell+1} - t_{\ell}
} 
\tilde{S}_\ell
(t_{\ell+1} - t_{\ell})
	\\&=
\sum_{\ell=m}^{n-1}
	\tilde{S}_\ell
	(
		\tilde{X}^i_{t_{\ell+1}}
			-
		\tilde{X}^i_{t_{\ell}}
	), 
\end{align*}
for some value
$
\tilde{S}_\ell
$
that satisfies 
\begin{equation}
\tilde{S}_\ell
	\in 
\{
	S(\tilde{X})_{t_m, s}^{w}
\mid
	t_{\ell} \le s \le t_{\ell+1}
\} 
\end{equation}
by the mean-value theorem.

Note that some of candidates of 
$
\tilde{S}_\ell
$
are
\begin{equation}
\label{eq:candidates}
	S(\tilde{X})_{t_m, t_{\ell}}^{w},
\quad
	\frac{1}{2} \big(
		S(\tilde{X})_{t_m, t_{\ell}}^{w}
	+
		S(\tilde{X})_{t_m, t_{\ell+1}}^{w}
	\big),
\quad
	S(\tilde{X})_{t_m, t_{\ell+1}}^{w} .
\end{equation}

\begin{defn}
\label{defn:HeadTail}
For the alphabet set
$I$,
we define the 
\newword{extended}
alphabet set
$\bar{I}$
by
\begin{equation}
\label{eq:extendedI}
\bar{I}
	:=
I \times
\{ -, + \} .
\end{equation}
For an alphabet
$i \in I$,
we call the extended alphabets
$i^- := (i, -) \in \bar{I}$
and
$i^+ := (i, +) \in \bar{I}$
the $\head$
and
the $\tail$
of
$i$,
respectively.
\end{defn}

In the following definition,
we will assign
the first and the last
candidates
in
(\ref{eq:candidates})
to
$\head$s
and
$\tail$s.


\begin{defn}{[Flat Discrete Signature]}
\label{defn:discSig}
Let
$I$
be an alphabet set,
$
X :
\mathcal{T}
	\to
\mathbb{R}^I
$
be
a discrete path,
and
$m, n \in \mathbb{N}$
with
$m < n$.
\begin{enumerate}
\item
For
$w \in \bar{I}^*$
and
$i \in I$,
$
S(X)_{t_m, t_n}^w
	\in
\mathbb{R}
$
is defined inductively by
\begin{align}
\label{eq:discSigBase}
S(X)_{t_m, t_n}^{\lambda}
	&:=
1,
	\\
\label{eq:discSigMinus}
S(X)_{t_m, t_n}^{w \otimes i^-}
	&:=
\sum_{\ell=m}^{n-1}
	S(X)_{t_m, t_{\ell}}^{w}
	(
		X^i_{t_{\ell+1}}
			-
		X^i_{t_{\ell}}
	) ,
	\\
\label{eq:discSigPlus}
S(X)_{t_m, t_n}^{w \otimes i^+}
	&:=
\sum_{\ell=m}^{n-1}
	S(X)_{t_m, t_{\ell+1}}^{w}
	(
		X^i_{t_{\ell+1}}
			-
		X^i_{t_{\ell}}
	) .
\end{align}

\item
The \newword{flat discrete signature} of
$X$
over
$[t_m, t_n]$
is a function
$
S(X)_{t_m, t_n}
	:
\bar{I}^* \to \mathbb{R}
$
defined by
for
$w \in \bar{I}^*$,
\begin{equation}
S(X)_{t_m, t_n}(w)
	:=
S(X)_{t_m, t_n}^w .
\end{equation}

\end{enumerate}
\end{defn}


\begin{prop}
\label{prop:discSigOneTwo}
For
$i, i_1, i_2, i_3 \in I$,
$* \in \{-, +\}$
and
$m, n \in \mathbb{N}$
with
$m < n$,
\begin{align}
S(X)_{t_m, t_n}^{i^*}
	&=
\sum_{m \le \ell < n}
	(
		X^{i}_{t_{\ell+1}}
			-
		X^{i}_{t_{\ell}}
	) 
	=
X^i_{t_n}
	-
X^i_{t_m},
	\\
\label{eq:SXminus)}
S(X)_{t_m, t_n}^{i_1^* \otimes i_2^-}
	&=
\sum_{m \le \ell_1 < \ell_2 < n}
	(
		X^{i_1}_{t_{\ell_1+1}}
			-
		X^{i_1}_{t_{\ell_1}}
	) 
	(
		X^{i_2}_{t_{\ell_2+1}}
			-
		X^{i_2}_{t_{\ell_2}}
	) ,
	\\
\label{eq:SXplus)}
S(X)_{t_m, t_n}^{i_1^* \otimes i_2^+}
	&=
\sum_{m \le \ell_1 \le \ell_2 < n}
	(
		X^{i_1}_{t_{\ell_1+1}}
			-
		X^{i_1}_{t_{\ell_1}}
	) 
	(
		X^{i_2}_{t_{\ell_2+1}}
			-
		X^{i_2}_{t_{\ell_2}}
	) ,
	\\
S(X)_{t_m, t_n}^{i_1^* \otimes i_2^- \otimes i_3^-}
	&=
\sum_{m \le \ell_1 < \ell_2 < \ell_3< n}
	(
		X^{i_1}_{t_{\ell_1+1}}
			-
		X^{i_1}_{t_{\ell_1}}
	) 
	(
		X^{i_2}_{t_{\ell_2+1}}
			-
		X^{i_2}_{t_{\ell_2}}
	) 
	(
		X^{i_3}_{t_{\ell_3+1}}
			-
		X^{i_3}_{t_{\ell_3}}
	) ,
	\\
S(X)_{t_m, t_n}^{i_1^* \otimes i_2^- \otimes i_3^+}
	&=
\sum_{m \le \ell_1 < \ell_2 \le \ell_3< n}
	(
		X^{i_1}_{t_{\ell_1+1}}
			-
		X^{i_1}_{t_{\ell_1}}
	) 
	(
		X^{i_2}_{t_{\ell_2+1}}
			-
		X^{i_2}_{t_{\ell_2}}
	) 
	(
		X^{i_3}_{t_{\ell_3+1}}
			-
		X^{i_3}_{t_{\ell_3}}
	) ,
	\\
S(X)_{t_m, t_n}^{i_1^* \otimes i_2^+ \otimes i_3^-}
	&=
\sum_{m \le \ell_1 \le \ell_2 < \ell_3< n}
	(
		X^{i_1}_{t_{\ell_1+1}}
			-
		X^{i_1}_{t_{\ell_1}}
	) 
	(
		X^{i_2}_{t_{\ell_2+1}}
			-
		X^{i_2}_{t_{\ell_2}}
	) 
	(
		X^{i_3}_{t_{\ell_3+1}}
			-
		X^{i_3}_{t_{\ell_3}}
	) ,
	\\
S(X)_{t_m, t_n}^{i_1^* \otimes i_2^+ \otimes i_3^+}
	&=
\sum_{m \le \ell_1 \le \ell_2 \le \ell_3< n}
	(
		X^{i_1}_{t_{\ell_1+1}}
			-
		X^{i_1}_{t_{\ell_1}}
	) 
	(
		X^{i_2}_{t_{\ell_2+1}}
			-
		X^{i_2}_{t_{\ell_2}}
	) 
	(
		X^{i_3}_{t_{\ell_3+1}}
			-
		X^{i_3}_{t_{\ell_3}}
	) .
\end{align}
\end{prop}
\begin{proof}
Straightforward.
\end{proof}

You may notice the correspondence between 
$\{-, +\}$
and
$\{<, \le \}$
in the ranges of summations in 
Proposition \ref{prop:discSigOneTwo}.

\begin{exmp}
\label{exmp:sampleOne}
Suppose that we observed 
2 dimensional data
in Table \ref{tbl:InDataStram}
with
$I = \{ 1, 2\}$.

\begin{table}[h]
\caption{Input data stream}
\label{tbl:InDataStram}
\begin{center}
\begin{tabular}{l|r|r|r|r|r}
$t$ & 0 & 1 & 1.5 & 2.5 & 3 \\ \hline
$X^1$ & 1 & 3 & & 5 & 8 \\
$X^2$ & 1 & 4 & 2 & & 6 \\ 
\end{tabular}
\end{center}
\end{table}

We will fill the missing data in 
in Table \ref{tbl:InDataStram}
with their latest values like the data
in
Table \ref{tbl:FilledDataStram}.

\begin{table}[h]
\caption{Filled data stream}
\label{tbl:FilledDataStram}
\begin{center}
\begin{tabular}{l|r|r|r|r|r}
$t$ & 0 & 1 & 1.5 & 2.5 & 3 \\ \hline
$X^1$ & 1 & 3 & \textbf{3} & 5 & 8 \\
$X^2$ & 1 & 4 & 2 & \textbf{2} & 6 \\ 
\end{tabular}
\end{center}
\end{table}

Then, 
the initial segment of the signature
$S(X)_{0,3}$
whose words length is less than or equal to 2,
has the following values,
where 
$* \in \{-, +\}$.
\begin{align*}
&S(X)_{0,3}^{\lambda}	= 1, \\
&S(X)_{0,3}^{1^*}	= 7, 
	\quad
&S(X)_{0,3}^{2^*}	= 5, \quad& \quad& \\
&S(X)_{0,3}^{1^*1^-}	= 16, 
	\quad
&S(X)_{0,3}^{1^*1^+}	= 33, 
	\quad
&S(X)_{0,3}^{1^*2^-}	= 12, 
	\quad
&S(X)_{0,3}^{1^*2^+}	= 30, \\
&S(X)_{0,3}^{2^*1^-}	= 5, 
	\quad
&S(X)_{0,3}^{2^*1^+}	= 23, 
	\quad
&S(X)_{0,3}^{2^*2^-}	= -2, 
	\quad
&S(X)_{0,3}^{2^*2^+}	= 27.
\end{align*}

\end{exmp}

In \cite{GLKF_2014},
the quadratic variation of 
any component of the original process
$X$
is provided by introducing
the lead-lag  transformation of streams.
Since 
the quadratic variation has a high relevance in financial applications,
this result was crucial.

The following theorem shows that
our $\head$-$\tail$ transformation also provides a similar functionality.

\begin{thm}
\label{thm:quandraticVal}
For
$i \in I$,
$* \in \{-, +\}$
and
$m, n \in \mathbb{N}$
with
$m < n$,
\begin{equation}
\label{eq:quandraticVal}
S(X)_{t_m, t_n}^{i^{*} \otimes i^+}
	-
S(X)_{t_m, t_n}^{i^{*} \otimes i^-}
	=
\sum_{m \le \ell < n}
	(
		X^{i}_{t_{\ell+1}}
			-
		X^{i}_{t_{\ell}}
	)^2  .
\end{equation}
\end{thm}
\begin{proof}
By
(\ref{eq:SXminus)})
and
(\ref{eq:SXplus)}),
we have
\begin{align*}
&
S(X)_{t_m, t_n}^{i^{*} \otimes i^+}
	-
S(X)_{t_m, t_n}^{i^{*} \otimes i^-}
	\\=&
\sum_{m \le \ell_1 \le \ell_2 < n}
	(
		X^{i}_{t_{\ell_1+1}}
			-
		X^{i}_{t_{\ell_1}}
	) 
	(
		X^{i}_{t_{\ell_2+1}}
			-
		X^{i}_{t_{\ell_2}}
	)
		-
\sum_{m \le \ell_1 < \ell_2 < n}
	(
		X^{i}_{t_{\ell_1+1}}
			-
		X^{i}_{t_{\ell_1}}
	) 
	(
		X^{i}_{t_{\ell_2+1}}
			-
		X^{i}_{t_{\ell_2}}
	)
	\\=&
\sum_{m \le \ell_1 = \ell_2 < n}
	(
		X^{i}_{t_{\ell_1+1}}
			-
		X^{i}_{t_{\ell_1}}
	) 
	(
		X^{i}_{t_{\ell_2+1}}
			-
		X^{i}_{t_{\ell_2}}
	)
	=
\sum_{m \le \ell < n}
	(
		X^{i}_{t_{\ell+1}}
			-
		X^{i}_{t_{\ell}}
	)^2  .
\end{align*}
\end{proof}

When applying signatures to time-series analysis,
we often encounter the necessity of treating data closer to the present time 
as more important than older data.
Let us think to generalize flat discrete signatures to reflect the fact.

Now 
for $m < n$,
we can rewrite 
(\ref{eq:discSigPlus}) 
as follows.
\begin{equation}
S(X)_{t_m, t_n}^{w \otimes i^+}
		=
\sum_{\ell=m}^{n-1}
	(
		X^i_{t_{\ell+1}}
			-
		X^i_{t_{\ell}}
	) 
	S(X)_{t_m, t_{\ell+1}}^{w}
		=
	\label{eq:SXdividedOldNew}
S(X)_{t_m, t_{n-1}}^{w \otimes i^+}
	+
(
	X^i_{t_{n}}
		-
	X^i_{t_{n-1}}
) 
S(X)_{t_m, t_{n}}^{w} .
\end{equation}
We can read
	(\ref{eq:SXdividedOldNew})
as
``First
$ S(X)_{t_m, t_{n-1}}^{w \otimes i^+} $
is computed at time $t_{n-1}$,
and then 
$(t_n - t_{n-1})$
later,
$ S(X)_{t_m, t_{n}}^{w} $
and
$ S(X)_{t_m, t_{n}}^{w \otimes i^+} $
are calculated using the (slightly outdated)
$ S(X)_{t_m, t_{n-1}}^{w \otimes i^+} $''.

Similarly,
we can rewrite
(\ref{eq:discSigMinus}) 
as follows.
\begin{equation}
S(X)_{t_m, t_n}^{w \otimes i^-}
	=
\sum_{\ell=m}^{n-1}
	(
		X^i_{t_{\ell+1}}
			-
		X^i_{t_{\ell}}
	) 
	S(X)_{t_m, t_{\ell}}^{w}
	=
	\label{eq:SXdividedOldNewMinus}
S(X)_{t_m, t_{n-1}}^{w \otimes i^-}
	+
(
	X^i_{t_{n}}
		-
	X^i_{t_{n-1}}
) 
S(X)_{t_m, t_{n-1}}^{w} .
\end{equation}
This time, we can read
	(\ref{eq:SXdividedOldNewMinus})
as
``First
$ S(X)_{t_m, t_{n-1}}^{w \otimes i^-} $
and
$ S(X)_{t_m, t_{n-1}}^{w} $
are computed at time $t_{n-1}$,
and then 
$(t_n - t_{n-1})$
later,
$ S(X)_{t_m, t_{n}}^{w \otimes i^+} $
is calculated using the (slightly outdated)
$ S(X)_{t_m, t_{n-1}}^{w \otimes i^-} $
and
$ S(X)_{t_m, t_{n-1}}^{w} $''.

In the following definition, a generalized version  of flat discrete signatures
is defined
by calculating
the outdated terms
with a weight of $1$ or less,
taking into account the elapsed time.

\begin{defn}{[discrete Signature]}
\label{defn:weightedSig}
Let
$I$
be an alphabet set,
$
X :
\mathcal{T}
	\to
\mathbb{R}^I
$
be
a discrete path,
$m, n \in \mathbb{N}$
with
$m < n$,
and
$\mu \ge 0$.
\begin{enumerate}
\item
For
$w \in \bar{I}^*$
and
$i \in I$,
$
S^{\mu}(X)_{t_m, t_n}^w
	\in
\mathbb{R}
$
is defined inductively by
\begin{align}
\label{eq:lambdaSigBase}
S^{\mu}(X)_{t_m, t_n}^{\lambda}
	&:=
1,
	\\
\label{eq:weightSigBase}
S^{\mu}(X)_{t_m, t_m}^{w}
	&:= \begin{cases}
		1
		&\textrm{if} \;
		w = \lambda,
			\\
		0
		&\textrm{otherwise},
\end{cases}
	\\
\label{eq:weightSigMinus}
S^{\mu}(X)_{t_m, t_n}^{w \otimes i^-}
	&:=
e^{- \mu (t_n - t_{n-1})}
\big(
S^{\mu}(X)_{t_m, t_{n-1}}^{w \otimes i^-}
	+
(X^i_{t_n} - X^i_{t_{n-1}})
S^{\mu}(X)^w_{t_m, t_{n-1}}
\big),
	\\
\label{eq:weightSigPlus}
S^{\mu}(X)_{t_m, t_n}^{w \otimes i^+}
	&:=
e^{- \mu (t_n - t_{n-1})}
S^{\mu}(X)_{t_m, t_{n-1}}^{w \otimes i^+}
	+
(X^i_{t_n} - X^i_{t_{n-1}})
S^{\mu}(X)^w_{t_m, t_{n}}.
\end{align}

\item
The \newword{discrete signature} of
$X$
with the decay rate
$\mu$
over
$[t_m, t_n]$
is a function
$
S^{\mu}(X)_{t_m, t_n}
	:
\bar{I}^* \to \mathbb{R}
$
defined by
for
$w \in \bar{I}^*$,
\begin{equation}
S^{\mu}(X)_{t_m, t_n}(w)
	:=
S^{\mu}(X)_{t_m, t_n}^w .
\end{equation}

\end{enumerate}
\end{defn}

Note that
$
S^{0}(X)_{t_m, t_n}
	=
S(X)_{t_m, t_n}.
$

\begin{prop}
\label{prop:weightedSigOneTwo}
For
$i, i_1, i_2 \in I$,
$m, n \in \mathbb{N}$
with
$m < n$,
and
$\mu > 0$,
\begin{align}
\label{eq:SmuXm)}
S^{\mu}(X)_{t_m, t_n}^{i^-}
	&=
\sum_{m \le \ell < n}
	e^{
		- \mu
		(t_n - t_{\ell})
	}
	(
		X^{i}_{t_{\ell+1}}
			-
		X^{i}_{t_{\ell}}
	) ,
	\\
\label{eq:SmuXp)}
S^{\mu}(X)_{t_m, t_n}^{i^+}
	&=
\sum_{m \le \ell < n}
	e^{
		- \mu
		(t_n - t_{\ell+1})
	}
	(
		X^{i}_{t_{\ell+1}}
			-
		X^{i}_{t_{\ell}}
	) ,
	\\
\label{eq:SmuXmm)}
S^{\mu}(X)_{t_m, t_n}^{i_1^- \otimes i_2^-}
	&=
\sum_{m \le \ell_1 < \ell_2 < n}
	e^{
		- \mu
		(t_{n} - t_{\ell_1})
	}
	(
		X^{i_1}_{t_{\ell_1+1}}
			-
		X^{i_1}_{t_{\ell_1}}
	) 
	(
		X^{i_2}_{t_{\ell_2+1}}
			-
		X^{i_2}_{t_{\ell_2}}
	) ,
	\\
\label{eq:SmuXmp)}
S^{\mu}(X)_{t_m, t_n}^{i_1^- \otimes i_2^+}
	&=
\sum_{m \le \ell_1 \le \ell_2 < n}
	e^{
		- \mu
		(t_{n} - t_{\ell_1})
	}
	(
		X^{i_1}_{t_{\ell_1+1}}
			-
		X^{i_1}_{t_{\ell_1}}
	) 
	(
		X^{i_2}_{t_{\ell_2+1}}
			-
		X^{i_2}_{t_{\ell_2}}
	) ,
	\\
\label{eq:SmuXpm)}
S^{\mu}(X)_{t_m, t_n}^{i_1^+ \otimes i_2^-}
	&=
\sum_{m \le \ell_1 < \ell_2 < n}
	e^{
		- \mu
		(t_{n} - t_{\ell_1+1})
	}
	(
		X^{i_1}_{t_{\ell_1+1}}
			-
		X^{i_1}_{t_{\ell_1}}
	) 
	(
		X^{i_2}_{t_{\ell_2+1}}
			-
		X^{i_2}_{t_{\ell_2}}
	) ,
	\\
\label{eq:SmuXpp)}
S^{\mu}(X)_{t_m, t_n}^{i_1^+ \otimes i_2^+}
	&=
\sum_{m \le \ell_1 \le \ell_2 < n}
	e^{
		- \mu
		(t_{n} - t_{\ell_1+1})
	}
	(
		X^{i_1}_{t_{\ell_1+1}}
			-
		X^{i_1}_{t_{\ell_1}}
	) 
	(
		X^{i_2}_{t_{\ell_2+1}}
			-
		X^{i_2}_{t_{\ell_2}}
	) .
\end{align}
\end{prop}
\begin{proof}
By induction on $n$.
\end{proof}

We have a similar result as 
Theorem \ref{thm:quandraticVal}
for discrete signatures,
which tells that
discrete signatures can represent ``weighted'' quadratic variations.
Actually, the result is a generalization of 
Theorem \ref{thm:quandraticVal}.

\begin{thm}
\label{thm:quandraticValMu}
For
$i \in I$,
and
$m, n \in \mathbb{N}$
with
$m < n$,
\begin{align}
\label{eq:quandraticValm}
S^{\mu}(X)_{t_m, t_n}^{i^- \otimes i^+}
	-
S^{\mu}(X)_{t_m, t_n}^{i^- \otimes i^-}
	&=
\sum_{m \le \ell < n}
	e^{
		-\mu
		(t_n - t_{\ell})
	}
	(
		X^{i}_{t_{\ell+1}}
			-
		X^{i}_{t_{\ell}}
	)^2  ,
\\
\label{eq:quandraticValp}
S^{\mu}(X)_{t_m, t_n}^{i^+ \otimes i^+}
	-
S^{\mu}(X)_{t_m, t_n}^{i^+ \otimes i^-}
	&=
\sum_{m \le \ell < n}
	e^{
		-\mu
		(t_n - t_{\ell+1})
	}
	(
		X^{i}_{t_{\ell+1}}
			-
		X^{i}_{t_{\ell}}
	)^2  .
\end{align}
\end{thm}
\begin{proof}
The proof is
exactly same as
that of
Theorem \ref{thm:quandraticVal},
by using
Proposition \ref{prop:weightedSigOneTwo}.
\end{proof}

\begin{exmp}
\label{exmp:sampleTwo}
Using the same data in Example \ref{exmp:sampleOne},
the initial segment of the discrete signature
$S^{\mu}(X)_{0,3}$
with the decay rate
$\mu = \log 2 \fallingdotseq 0.693$ (half-life $= 1$)
whose words length is less than or equal to 2,
has the following values.
\begin{align*}
&S^{\mu}(X)_{0,3}^{\lambda}	= 1,  \\
&S^{\mu}(X)_{0,3}^{1^-}	= 3.08, 
	\quad
&S^{\mu}(X)_{0,3}^{1^+}	= 4.91, 
	\quad
&S^{\mu}(X)_{0,3}^{2^-}	= 2.70, 
	\quad
&S^{\mu}(X)_{0,3}^{2^+}	= 4.04, \\
&S^{\mu}(X)_{0,3}^{1^-1^-}	= 3.37, 
	\quad
&S^{\mu}(X)_{0,3}^{1^-1^+}	= 11.65, 
	\quad
&S^{\mu}(X)_{0,3}^{1^-2^-}	= 3.33, 
	\quad
&S^{\mu}(X)_{0,3}^{1^-2^+}	= 12.56, \\
&S^{\mu}(X)_{0,3}^{1^+1^-}	= 6.74, 
	\quad
&S^{\mu}(X)_{0,3}^{1^+1^+}	= 19.57, 
	\quad
&S^{\mu}(X)_{0,3}^{1^+2^-}	= 6.66, 
	\quad
&S^{\mu}(X)_{0,3}^{1^+2^+}	= 20.16, \\
&S^{\mu}(X)_{0,3}^{2^-1^-}	= -0.63, 
	\quad
&S^{\mu}(X)_{0,3}^{2^-1^+}	= 8.61, 
	\quad
&S^{\mu}(X)_{0,3}^{2^-2^-}	= -1.25, 
	\quad
&S^{\mu}(X)_{0,3}^{2^-2^+}	= 12.19, \\
&S^{\mu}(X)_{0,3}^{2^+1^-}	= 0.21, 
	\quad
&S^{\mu}(X)_{0,3}^{2^+1^+}	= 13.71, 
	\quad
&S^{\mu}(X)_{0,3}^{2^+2^-}	= -1.33, 
	\quad
&S^{\mu}(X)_{0,3}^{2^+2^+}	= 18.34. 
\end{align*}

\end{exmp}

\section{An implementation of discrete signature}
\label{sec:ImpleDS}

\newcommand{\sigpy}{\mathbf{sig.py}}
\newcommand{\Data}{\mathbf{Data}}
\newcommand{\Words}{\mathbf{Words}}
\newcommand{\Signature}{\mathbf{Signature}}

\lstdefinestyle{myCustomPythonStyle}{
	language=Python,
	numbers=left,
	stepnumber=1,
	numbersep=10pt,
	tabsize=2,
	showspaces=false,
	showstringspaces=false
}

\lstdefinestyle{myCustomDataStyle}{
	language=Python,
	stepnumber=1,
	numbersep=10pt,
	tabsize=8,
	showspaces=false,
	showstringspaces=false
}

In this section,
we will make a brief description about an implementation of 
discrete signature with Python
\cite{beazley2022}.
You can see
the whole code
$\sigpy$
and the data
$\mathbf{sample1.dat}$
 in
Table \ref{tbl:InDataStram}
at
\url{https://github.com/takanori-adachi/discrete-signature}.

Let us explain the functionality of classes in 
$\sigpy$
in the following subsections.

\subsection{The class $\Data$}
\label{sec:classData}

Suppose we have a data stream like
the following tab-separated records,
 which is corresponding to
the data in
Table \ref{tbl:InDataStram}.

\lstset{basicstyle=\small,style=myCustomDataStyle}
\begin{lstlisting}
;time	event_type	value
0.0	1	1.0
0.0	2	1.0
1.0	1	3.0	
1.0	2	4.0	
1.5	2	2.0
2.5	1	5.0
3.0	1	8.0
3.0	2	6.0
\end{lstlisting}

The class $\Data$ will perform the conversion from the above data stream to the filled data
specified in Table \ref{tbl:FilledDataStram}.
It reads the input stream (raw data) from a file and stores it into a list
\texttt{self.raw\_data}
Then, collects the elements of 
$I$ (the set of event types, \texttt{self.I}), 
$\bar{I}$ (the set of extended event types, \texttt{self.barI})
and
$\mathcal{T} (time domain, \texttt{self.T}$,
converting them into the internal integer values,
and preparing dictionaries for the conversions.
It finally creates $I$-dimensional discrete path $X$, or \texttt{self.X}.

The method
\texttt{w2mi}
converts a word to a list of integers representing alphabets,
or elements of $I$
containing in the word.
Conversely, 
\texttt{mi2w}
converts a list of integers to the corresponding word.
The data member
\texttt{t2i}
is the dictionary converting from an actual time to its corresponding index.

\subsection{The class $\Words$}
\label{sec:classWards}

The class $\Words$ generates the set
$I^{\le k}$
as a list of its elements (words).
The resulting list of elements
of the set
$I^{\le k}$
is stored in
the data member 
\texttt{self.Istar}.

In the flat case,
i.e. when
$\mu = 0$,
we have
\begin{equation}
S(X)_{t_m, t_n}^{i^- \otimes w}
	=
S(X)_{t_m, t_n}^{i^+ \otimes w}
\end{equation}
for
$i \in I$
and
$w \in \bar{I}$
by Proposition \ref{prop:discSigOneTwo}.
Therefore, we can identify
$i^-$
and
$i^+$
for the first alphabet
$i \in I$.
So, we prepare a separate universe of words
\texttt{self.IstarHalf}
 for the case
$\mu = 0$.

\subsection{The class $\Signature$}
\label{sec:classSignature}

A signature is initialized with a
$\mathbf{Data}$ object \texttt{data}
and the maximum length of words \texttt{k}.
The class
$\Signature$
encapsulate the heart of the computation of discrete signatures.

\lstset{basicstyle=\small,style=myCustomPythonStyle}
\begin{lstlisting}
class Signature(object): # discrete signature
	def __init__(self, data, k):
		self.data = data
		self.k = k # maximum length of words

	def sig(self, t1, t2, w):
		return(self.sig0(data.t2i[t1], data.t2i[t2], data.w2mi(w)))
	
	def sig0(self, m, n, iss):
		v = 1.0
		if len(iss) > 0:
			w = iss[:-1]
			i = iss[len(iss)-1]
			j, s = self.i2js(i)
			if s == 0: # HEAD
				v = self.mu_delta_t[n-1] * (self.sig0(m, n-1, iss)
					+ data.delta_X[n-1,j] * self.sig0(m, n-1, w))
			else: # TAIL
				v = self.mu_delta_t[n-1] * self.sig0(m, n-1, iss)
					+ data.delta_X[n-1,j] * self.sig0(m, n, w)
		return(v)
\end{lstlisting}
where
\texttt{mu\_delta\_t[n]}
is
$
e^{
	- \mu
	(t_{n+1} - t_n)
}
$,
and
\texttt{delta\_X[n-1,j]}
is
a data member defined in the class $\Data$
as 
$
X_{t_{n+1}}^j - X_{t_n}^j
$.
The function
\texttt{i2js}
converts a given index specifying an element of
$\bar{I} = I \times \{+, -\}$
to a pair
\texttt{(j, s)}
where 
$j \in I$
and 
$ s = 0 $
if
$i = j^-$,
and
$ s = 1 $
if
$i = j^+$.

The function
\texttt{sig}
simply calls
another function
\texttt{sig0}
after converting its arguments to corresponding internal representations.
The function
\texttt{sig0}
is a straightforward implementation of equations
(\ref{eq:lambdaSigBase}),
(\ref{eq:weightSigBase}),
(\ref{eq:weightSigMinus})
and
(\ref{eq:weightSigPlus}).
Note that it uses the recursive call technique.

This simple implementation, however,
is not so efficient.
In fact,
in the recursive call of
the function 
\texttt{sig0},
it repeats computations many times
for the same arguments,
which is simply a waste of time.

In order to avoid this extra computation,
we will introduce a container object
\texttt{Signature.v}
for holding results of computation so far.


First, we introduce the container object
\texttt{Words.v}.
It
consists of 
binary and multinary tree structures.
For each word
$w \in \bar{I}^{\le k}$,
we have a pair
\begin{equation}
c_w 
	:=
(b_w, r_w),
\end{equation}
where 
$r_w$
is the value of the signature at $w$, 
and $b_w$ is a boolean value 
that indicates whether $r_w$ has been calculated or not.
The intermediate container
$v_w$
is defined by the following recursive definition.
\begin{align*}
v_{w} &:= (c_w, (
	v_{w \otimes i_1},
		\cdots
	v_{w \otimes i_{\bar{d}}}
)),
		&(\textrm{for } w \in \bar{I}^{\le {(k-1)}})
		\\
v_{w} &:= (c_w, ()),
		&(\textrm{for } w \in \bar{I}^k)
\end{align*}
where
$\bar{d}$
is the cardinality of
$\bar{I}$
and
$
\{
i_1, \cdots, i_{\bar{d}}
\}
	=
\bar{I}
$.
Then,
the container
\texttt{Words.v}
is defined by
$v_{\lambda}$.

Next, we construct a container
\texttt{Signature.v}
which is a double list of 
\texttt{Words.v}.
For each pair of time
$(t_m, t_n) \in \mathcal{T}$
with
$m < n$.
%
The function \texttt{Signature.get\_c}
retrieves 
$c_w$ 
from 
$v_{m,n}$
for the word 
$w$
whose index is \texttt{iss}.
Using
\texttt{Signature.v},
the function
\texttt{Signature.sig0}
can be rewritten as:
\lstset{basicstyle=\small,style=myCustomPythonStyle}
\begin{lstlisting}
	def sig0(self, m, n, iss): # faster algorithm using container self.v
		c = self.get_c(m, n, iss)
		if c[0]: # if already computed
			return c[1] # return its value
		# otherwise, compute from scratch
		v = 1.0
		if len(iss) > 0:
			w = iss[:-1]
			i = iss[len(iss)-1]
			j, s = self.i2js(i)
			if s == 0: # HEAD
				v = self.mu_delta_t[n-1] * (self.sig0(m, n-1, iss) 
					+ data.delta_X[n-1,j] * self.sig0(m, n-1, w))
			else: # TAIL
				v = self.mu_delta_t[n-1] * self.sig0(m, n-1, iss) 
					+ data.delta_X[n-1,j] * self.sig0(m, n, w)
		c[0] = True # it is computed
		c[1] = v # and its value is 'v'
		return(v)
\end{lstlisting}

We will use this faster version in Section
\ref{sec:appFin}.

\section{An application of discrete signature to finance}
\label{sec:appFin}

As an example of applications of discrete signature to finance,
in this section,
we consider the problem of judging 
whether a given price-shares process is of the morning or of the afternoon session in 
Tokyo Stock Exchange (TSE).
TSE has morning (9:00-11:30) and afternoon (12:30-15:00) sessions each trading day.
Therefore, each session has 2 hours and 30 minutes.

We use FLEX Full historical data bought from TSE as the raw data.
FLEX Full data consists of high frequency tick data from which 
we can extract several micro dynamic data such as \textit{ita} data
or limit order book data.
The time resolution of FLEX Full data is currently 1 microsecond, or 
$10^{-6}$ second.
In the following, time is displayed in minutes. For example,
"09:12:34.567890"
is represented by the value
$
9 \times 60 + 12 + 34.567890 / 60
	=
552.5761315
$.

\subsection{Make a one-minute interval data stream}
\label{sec:makeOneMinDS}

We extract data stream 
\begin{equation}
\mathcal{D} =
\{ D_t \}
\end{equation}
from FLEX Full data, where 
$t$
is an observed time in minutes,
and,
each 
$D_t$
consists of the following five components:
\begin{align*}
D_t.P^a
	\quad &- \quad
\textrm{best ask price},
	\\
D_t.P^b
	\quad &- \quad
\textrm{best bid price},
	\\
D_t.S^a
	\quad &- \quad
\textrm{the total of ask side shares},
	\\
D_t.S^b
	\quad &- \quad
\textrm{the total of bid side shares},
	\\
D_t.V
	\quad &- \quad
\textrm{accumulated execution volume}.
\end{align*}

We will generate a substream
of
$\{ D_t \}$
at one-minute interval for each trading session.

First, let us define index sets of one minute interval blocks from the original data
by
for
$n \in \mathbb{N} := \{0, 1, 2, \cdots \}$,
\begin{align}
J_n
	&:=
\{
	t
\mid
	n \le t < n+1
		\; \textrm{and} \;
	D_t \in \mathcal{D}
\},
	\\
\bar{J}_n
	&:=
\{
	t
\mid
	n \le t \le n+1
		\; \textrm{and} \;
	D_t \in \mathcal{D}
\}.
\end{align}

Next, define pairs of times denoting open and close times of the session.
\begin{align}
(N_0, N_1)
	&\in
\{
(9 \times 60, 11.5 \times 60),
(12.5 \times 60, 15 \times 60)
\},
		\\
N
	&:=
N_1 - N_0
	=
150.
\end{align}

If 
$
D_{t_{\max J_{N_0}}}.V = 0
$,
i.e.
the security had not been open
in the first minute of the session,
we do not use the session as data and throw it away.
By assuming
$
D_{t_{\max J_{N_0}}}.V > 0
$,
we pick $D_t$ for each 
$n = N_0, N_0+1, \cdots, N_1$,
which is called
$\bar{D}_n$,
by the following procedure:
\begin{align*}
	&
\bar{D}_{N_0}
	:=
D_{\min J_{N_0}}
	\\&
\mathbf{for}\; n \;\mathbf{in}\; \mathbf{range}(N_0+1, N_1):
	\\& \quad
	\mathbf{if}\; J_{n-1} = \emptyset:
	\quad \quad
		\bar{D}_{n} := \bar{D}_{n-1}
	\\& \quad
	\mathbf{else}:
	\quad \quad
		\bar{D}_{n} := D_{\max J_{n-1}}
\\&
\mathbf{if}\; \bar{J}_{N_1-1} = \emptyset:
\quad
	\bar{D}_{N_1} := \bar{D}_{N_1 -1}
\\& 
\mathbf{else}:
\quad
	\bar{D}_{N_1} := D_{\max \bar{J}_{N_1-1}}
\end{align*}
Then, we got a one-minute interval data stream
\begin{equation}
\{
	\bar{D}_n
\}_{n = N_0, \cdots, N_1}
\end{equation}
for each session.

\subsection{Time normalization}
\label{sec:timeNormalization}
Since our problem is to detect time-related information of the given data stream,
we will eliminate clues by normalizing the time.
The followings are normalized time and its corresponding components.
For
$n = 0, 1, \cdots, N$,
\begin{align}
t_n
	&:=
\frac{n}{N},
	\\
P^a_{t_n}
	&:=
\bar{D}_{N_0+n}.P^a,
	\\
P^b_{t_n}
	&:=
\bar{D}_{N_0+n}.P^b,
	\\
S^a_{t_n}
	&:=
\bar{D}_{N_0+n}.S^a,
	\\
S^b_{t_n}
	&:=
\bar{D}_{N_0+n}.S^b,
	\\
V_{t_n}
	&:=
\bar{D}_{N_0+n}.V.
\end{align}
Then, our time domain is
\begin{equation}
\mathcal{T}
	:=
\{
	t_0, t_1, \cdots, t_{N}
\}.
\end{equation}

\subsection{Make a discrete path for each session}
\label{sec:makeDiscPath}

We introduce some other statistics.
For 
$t \in \mathcal{T}$,
\begin{align}
p_{t}
	&:=
\ln \frac{
	P^a_t
		+
	P^b_t
}{2},
	&(\textrm{logarithm of mid-price})
	\\
s_t
	&:=
P^a_t
	-
P^b_t.
	&(\textrm{spread})
\end{align}

Next, we construct a discrete path
\begin{equation}
X := (X^1, X^2, X^3, X^4) : \mathcal{T} \to \mathbb{R}^I
\end{equation}
with
\begin{equation}
\label{eq:Ifour}
I := \{1,2,3,4\}
\end{equation}
from which we will compute its discrete signature.
For 
$t \in \mathcal{T}$,
\begin{align}
X^1_t
	&:=
\frac{
	p_t - \single{p}
}{
	\sqrt{\single{p^2} - \single{p}^2}
},
	&(\textrm{normalized logarithm of mid-price})
	\\
X^2_t
	&:=
\frac{
	s_t - \single{s}
}{
	\sqrt{\single{s^2} - \single{s}^2}
},
	&(\textrm{normalized spread})
	\\
X^3_t
	&:=
\frac{
	S^a_t - S^b_b
}{
	S^a_t + S^b_b
},
	&(\textrm{normalized imbalance})
	\\
X^4_t
	&:=
\frac{
	V_t
}{
	V_1
},
	&(\textrm{normalized accumulated volume})
\end{align}
where
$
\single{x}
	:=
\frac{1}{N}
\sum_{t \in \mathcal{T}} x_t
$
for any sequence
$
	\{x_t\}_{t \in \mathcal{T}}
$.

\subsection{Experiment and Result}
\label{sec:expResult}
In the experiment,
we used data from 
January 2020 to July 2021
for 30 names in TOPIX CORE 30.
After shuffling date,
we use 80\% of the whole data for training, 
and use 20\% for test.

The calculated signature is used to determine the morning and afternoon sessions
using logistic regression
by which
a binary decision was made, with 0 for the morning and 1 for the afternoon.

The set of event types or statistics is $I$ defined in
(\ref{eq:Ifour}).
We pick the seven sorts of feature sets as subsets of
$\bar{I}^{\le k}$,
$\overline{\{1\}}^{\le k}$,
$\overline{\{2\}}^{\le k}$, 
$\overline{\{3\}}^{\le k}$, 
$\overline{\{4\}}^{\le k}$, 
$\bar{I}^{\le k}$ itself,
$\overline{\{2,4\}}^{\le k}$ 
and
$
\{
	w \in \bar{I}^{\le k}
\mid
	w \sim /[4^-4^+]/
\}
$,
where
``$w \sim /[4^-4^+]/$''
means
``$w$ matches the pattern 
$[4^-4^+]$''.
In other words, it means
``$w$ contains the (extended) alphabets $4^-$ or $4^+$''.
We check these patterns for
$k=1, 2, 3$.

Table \ref{tbl:withSig}
shows the accuracy of logistic regression adopting 
members of the feature set as its explanatory variables.

\begin{table}[hbtp]
\caption{Computation with Signature}
\label{tbl:withSig}
\centering
\begin{tabular}{|c|r|r|r|r|r|r|}
\hline
\multirow{2}{*}{Feature set} & \multicolumn{3}{r|}{Accuracy} & \multicolumn{3}{r|}{Number of features} \\
\cline{2-7}
& $k$=1 & $k$=2 & $k$=3 & $k$=1 & $k$=2 & $k$=3 \\
\hline
$\overline{\{1\}}^{\le k}$  
& 50.72\% & 55.46\% & 55.91\% & 1 & 3 & 7 \\
$\overline{\{2\}}^{\le k}$  
& 72.51\% & 75.54\% & 83.08\% & 1 & 3 & 7 \\
$\overline{\{3\}}^{\le k}$  
& 55.04\% & 58.96\% & 59.14\% & 1 & 3 & 7 \\
$\overline{\{4\}}^{\le k}$  
& 90.18\% & 93.01\% & 97.46\% & 1 & 3 & 7 \\
$\bar{I}^{\le k}$  
& 89.63\% & 98.86\% & 99.51\% & 4 & 36 & 292 \\
$\overline{\{2,4\}}^{\le k}$  
& 89.58\% & 98.84\% & 99.82\% & 2 & 10 & 42 \\
$
\{
	w \in \bar{I}^{\le k}
\mid
	w \sim /[4^-4^+]/
\}
$  
& 90.18\% & 97.84\% & 99.55\% & 1 & 15 & 163 \\
\hline
\end{tabular}
\end{table}

The statistics ``4'' 
(normalized cumulative volume) 
apparently made the best performance, 
and the statistics ``2''
(normalized spread) 
is next.
That is why we tried the 
$\{2,4\}$ case and 
the last case that treats only words containing ``4''.
You may see that
the values of the sixth and the last cases are better than that of the whole set
$\bar{I}^{\le k}$
case
at $k=3$,
while the number of features of the 
$\{2,4\}$ case and 
the last case 
are much less than the whole set case.

\lspace

Let us mention the computation speed of obtaining the signature in 
Table \ref{tbl:withSig}.
%
The workstation we used for the computation has 2 CPUs.
Each CPU has 48 cores,
and each core can handle 2 threads.
So, the total number of threads is 192, 
which is the number of affordable distributed parallel processing.
We used 150 threads out of 192 for our computation in order to avoid overwhelming the tasks of other users.
The computation of all components of 
$\bar{I}^{\le 4}$ of the signature took 
68 minutes and 33.266 seconds.


\lspace

In order to evaluate the result in 
Table \ref{tbl:withSig}
fairly,
we also performed logistic regression using the raw data as it is without using signature
as a comparison.
The result is shown in 
Table \ref{tbl:withoutSig}.

\begin{table}[hbtp]
\caption{Computation without Signature}
\label{tbl:withoutSig}
\centering
\begin{tabular}{|l|r|r|}
\hline
Statistics & Accuracy & Number of features \\
\hline
Normalized logarithm of mid-price & 60.14\% & 151 \\
Normalized spread & 88.18\% & 151 \\
Normalized imbalance & 66.90\% & 151 \\
Normalized cumulative volume & 99.73\% & 151 \\
All & 99.64\% & 604 \\
\hline
\end{tabular}
\end{table}

One of the most important points in the comparison is
the number of features required to achieve good accuracy.
For example, in 
$k=3$ cases,
the logistic regression using all raw 604 data points
performs 99.64\% accuracy
while
the logistic regression using 42 components of the discrete signature
specified by
$\overline{\{2,4\}}^{\le 3}$  
performs 99.82\% accuracy
which is slightly better than the former case.
In other words,
the regression with the feature set specified by the signature can achieve almost the same level of good results
as the regression with the whole raw data set
with much fewer data points.

\section{Concluding Remarks}
\label{sec:conclusion}

We would like to leave a few remarks before finishing this paper.

The lead-lag transformation needs to double the cardinality $n$ of the time domain $\mathcal{T}$
while our $\head$-$\tail$ transformation needs to double the cardinality $d$ of the alphabet set $I$.
Then, the ratio of computation times of these two methods will be
$
\frac{
	d^{2 n}
}{
	(2 d)^n
}
	=
\big(
	\frac{d}{2}
\big)^n
$.
Therefore,
the lead-lag transformation
will take 
more time than ours
when
$d > 2$.

We used the pattern
``$[4^-4^+]$''
in Table \ref{tbl:withSig}
for specifying the subset of
$\bar{I}^{\le k}$.
In general,
a subset of
$I^*$
is called a language 
in Mathematical Language Theory 
\cite{sipser2013}.
There are some popular languages in this sense 
including regular languages and context-free languages.
By modifying the class $\Words$
with the Python built-in library \textbf{re}, 
we can easily extend it to handle regular languages, 
i.e. languages generated by regular expressions.
This gives us a more possibility to specify smaller and more appropriate feature sets
instead of using whole
$\bar{I}^{\le k}$  
whose cardinality is 292 when
$k = 3$
in Section \ref{sec:appFin}.




\lspace
\noindent
\textbf{Acknowledgement:}
We would like to thank to 
Research Center for Quantitative Finance, Tokyo Metropolitan University
for allowing us to use their high-speed workstation, \textit{Turing}.





\end{document}